\documentclass[onecolumn]{autart}    

\usepackage{graphicx}          

\usepackage[english]{babel}
\usepackage[utf8x]{inputenc}
\usepackage[T1]{fontenc}

\usepackage{ulem}

\usepackage{amsmath}
\usepackage{amssymb}
\usepackage{graphicx}
\usepackage[colorlinks=true, allcolors=blue]{hyperref}

\usepackage{algorithm}
\usepackage[noend]{algpseudocode}

\newcounter{comments}
\usepackage{xcolor}
\definecolor{Red}{rgb}{0.8,0,0}
\definecolor{Green}{rgb}{0,0.7,0}

\algrenewcommand\algorithmicforall{\textbf{foreach}}
\algrenewcommand\algorithmicindent{.8em}

\newtheorem{theorem}{Theorem}[section]   

\newtheorem{corollary}[theorem]{Corollary}     
\newtheorem{proposition}[theorem]{Proposition}  

\theoremstyle{remark}
\newtheorem{definition}[theorem]{Definition}   

\newtheorem{remark}[theorem]{Remark}        
\newtheorem{example}[theorem]{Example}        
\newtheorem{proof}[theorem]{Proof}        

\begin{document}

\begin{frontmatter}

\title{Revealing the canalizing structure of Boolean functions: Algorithms and applications}

\thanks[footnoteinfo]{Corresponding author David Murrugarra.}

\author[Elena]{Elena Dimitrova}\ead{edimitro@calpoly.edu},    
\author[Brandy]{Brandilyn Stigler}\ead{bstigler@smu.edu},      
\author[Claus]{Claus Kadelka}\ead{ckadelka@iastate.edu},  
\author[David]{David Murrugarra}\ead{murrugarra@uky.edu}

\address[Elena]{Department of Mathematics, California Polytechnic State University, San Luis Obispo, CA, USA}  
\address[Brandy]{Department of Mathematics, Southern Methodist University, Dallas, TX, USA}             
\address[Claus]{Department of Mathematics, Iowa State University, Ames, IA, USA}        
\address[David]{Department of Mathematics, University of Kentucky, Lexington, KY, USA}
          
\begin{keyword}                           
Boolean functions, Canalizing layers, \textbf{NP}-hard, Disjunctive normal forms, Reverse engineering.               
\end{keyword}                             

\begin{abstract}                          
Boolean functions can be represented in many ways including logical forms, truth tables, and polynomials.
Additionally, Boolean functions have different canonical representations such as minimal disjunctive normal forms.
Another canonical representation is based on the polynomial representation of Boolean functions and the biologically motivated concept of canalization: any Boolean function 
can be written as a nested product of canalizing layers and a polynomial that contains the variables that are never canalizing.
In this paper we study the problem of identifying the canalizing layers of a Boolean function.
First, we show that the problem of finding the canalizing layers is \textbf{NP}-hard.
Second, we present several algorithms for finding the canalizing layers, discuss their complexities, and compare their performances.
Third, we exhibit how the canalizing layers format can
be used to find a disjunctive normal form for any nested canalizing function.
Another application deals with the reverse engineering of Boolean networks with a prescribed layering format. Implementations of the developed algorithms in Python and in the computer algebra system Macaulay2 are available at \href{https://github.com/ckadelka/BooleanCanalization}{https://github.com/ckadelka/BooleanCanalization}.
\end{abstract}

\end{frontmatter}
\normalem

\section{Introduction}
First introduced in the 1940s, canalization refers to a concept from evolutionary and developmental biology used to describe the ability of organisms to produce and maintain the same phenotype in the face of widespread variability in their genotypes~\cite{waddington1942canalization,Siegal:2002wz}. S. Kauffman introduced Boolean canalizing functions and studied their role in 
gene regulatory network models~\cite{kauffman1974large}. 
Over the years, more  evidence emerged that canalization is ubiquitous in gene regulation~\cite{harris2002model,daniels2018criticality}: most rules used in over 130 expert-curated published gene regulatory network models are even so-called nested canalizing~\cite{kadelka2020meta}. These findings intensified the study of Boolean canalization and its impact on the dynamics and controllability of gene regulatory network models~\cite{jarrah2007nested,Murrugarra:2015uq}. 
Boolean networks governed by canalizing functions are typically more stable than networks governed by random functions~\cite{kauffman2004genetic,kadelka2017influence}. That is, they have fewer and shorter attractors and are more robust to perturbations. The amount of canalization exhibited by a Boolean function can be quantified using concepts such as the canalizing depth~\cite{Layne:2012aa,he2016stratification}, the canalizing strength~\cite{kadelka2020collectively}, or the logical reduncancy~\cite{marques2013canalization,gates2021effective}. As a rule of thumb, the higher the amount and prevalence of canalization the more stable are the dynamics~\cite{karlsson2007order}.

A canalizing function possesses at least one ``canalizing'' variable such that, if this variable takes on its ``canalizing'' value, then the output value is already determined, regardless of the values of the remaining variables. 
If the canalizing variable does not take on its canalizing value, but there is a second variable with the canalization property, the function is $2$-canalizing. If $k$ variables follow this pattern, the function is $k$-canalizing.  If all variables are canalizing in this iterative way, the function is ``nested canalizing''; the Boolean AND function is such an example. 
Boolean nested canalizing functions have been shown to be equivalent to unate cascade functions~\cite{jarrah2007nested}, a class of functions that have been extensively studied by computer scientists~\cite{Maitra1962,Sasao_1979}. The concept of canalization has also been extended to the multistate case where variables can take on three or more input values and the dynamics of networks governed by multistate canalizing functions have been studied~\cite{Murrugarra:2011aa,Murrugarra:2012aa,kadelka2017multistate,MurrugarraD21}.

The concept of canalizing layers was first introduced for nested canalizing functions~\cite{Li:2013aa}. Later, a full stratification of the variables of any Boolean function into layers was provided~\cite{he2016stratification}. This layer structure provides a unique polynomial form (up to permutations within the layers) of representing any Boolean function and contains various important information about the dynamical properties of the function. 
For instance, the canalizing layers of a nested canalizing function are related to the function's symmetry groups~\cite{rosenkrantz2019symmetry}; all variables in the same layer are contained within at most two symmetry groups. Further, the layer structure has been used for quantifying the total effect of network perturbations and control~\cite{Murrugarra:2015uq,MurrugarraD21}, as well as for calculating exact Derrida values, a stability measure of Boolean networks, for networks based on nested canalizing functions~\cite{kadelka2017influence}.


In this paper we present multiple algorithms for finding the canalizing layers of a Boolean function. We also provide implementations of the algorithms in Python and in the computer algebra system Macaulay2~\cite{M2}. Although our implementations can handle finding the layer structure for most functions in published models~\cite{kadelka2020meta}, the problem is \textbf{NP}-hard, as we prove in Section~\ref{sec:np-hard}.

The canalizing layer structure can be used to obtain useful information about the Boolean functions in a model. In this paper, we focus on two practical applications.
First, we show how to obtain a disjunctive normal form for a nested canalizing function from the  layers. Second, we provide an application for how to reverse engineer a family of Boolean functions with a prescribed layering format.

This paper is structured in the following way. We provide the theoretical background of Boolean canalization and the layer structure in Section~\ref{sec:background} and show that the problem of finding the canalizing layers is \textbf{NP}-hard in Section~\ref{sec:np-hard}.
Then in Section~\ref{sec:algorithms} we present multiple algorithms along with their performance and conclude with applications. We describe how to obtain disjunctive normal forms from canalization layers of Boolean functions in Section~\ref{sec:DNF}, and show how to use layers for reverse engineering in Section~\ref{sec:rev-eng}. In Section~\ref{sec:discussions} we briefly summarize the main findings and their significance.

\section{Background}
\label{sec:background}
In this section, we review some concepts and definitions and introduce the concept of \emph{canalization} as well as \emph{layers of canalization}. Throughout the manuscript, we consider functions over a two-element field $\mathbb{F}_2 = \{0,1\}$ (i.e., Boolean functions). We frequently use the fact that any Boolean function $f \colon \mathbb{F}_2 \times \cdots \times \mathbb{F}_2 \to \mathbb{F}_2$ can be represented by a polynomial function with coefficients in $\mathbb{F}_2$, which enables the use of tools from algebraic geometry.

\begin{definition}\label{def:essential}
Given a Boolean function $f(x_1,\ldots,x_n)$, variable $x_i$ is \emph{essential} if the two functions $$f(x_1,\ldots,x_{i-1},0,x_{i+1},\ldots,x_n)$$ and $$f(x_1,\ldots,x_{i-1},1,x_{i+1},\ldots,x_n)$$ are not equivalent. Otherwise, $x_i$ is \emph{non-essential}.
\end{definition}

\begin{definition}\label{def:canalizing}
A Boolean function $f(x_1,\ldots,x_n)$ is \emph{canalizing} if there exists a variable $x_i$, values $a,\,b\in\{0,\,1\}$, and a Boolean function $g(x_1,\ldots,x_{i-1},x_{i+1},\ldots,x_n)\not\equiv b$  such that
$$f(x_1,\ldots,x_n)= \begin{cases}
b, \text{if}\ x_i=a\\
g(x_1,\ldots,x_{i-1},x_{i+1},...,x_n), \text{if}\ x_i\neq a.
\end{cases}$$
In this case, $x_i$ is called a \emph{canalizing variable}, $a$ its \emph{canalizing input}, and $b$ the \emph{canalized output}.

\end{definition}

\begin{definition}[\cite{he2016stratification}]\label{def:kcanalizing} 
A Boolean function $f(x_1,\ldots,x_n)$ is \emph{$k$-canalizing}, where $1 \leq k \leq n$, with respect to the permutation $\sigma \in \mathcal{S}_n$, inputs $a_1,\ldots,a_k$, and outputs $b_1,\ldots,b_k$ if

$f(x_{1},\ldots,x_{n})=$
\begin{equation*}
\left\{\begin{array}[c]{ll}
b_{1} & x_{\sigma(1)} = a_1,\\
b_{2} & x_{\sigma(1)} \neq a_1, x_{\sigma(2)} = a_2,\\
b_{3} & x_{\sigma(1)} \neq a_1, x_{\sigma(2)} \neq a_2, x_{\sigma(3)} = a_3,\\
\vdots  & \vdots\\
b_{k} & x_{\sigma(1)} \neq a_1,\ldots,x_{\sigma(k-1)}\neq a_{k-1}, x_{\sigma(k)} = a_k,\\
f_C\not\equiv b_k & x_{\sigma(1)} \neq a_1,\ldots,x_{\sigma(k-1)}\neq a_{k-1}, x_{\sigma(k)} \neq a_k,
\end{array}\right.\end{equation*}
where $f_C = f_C(x_{\sigma(k+1)},\ldots,x_{\sigma(n)})$ is the \emph{core function}, a Boolean function on $n-k$ variables. When $f_C$ is not canalizing, then the integer $k$ is the \emph{canalizing depth} of $f$ \cite{Layne:2012aa}. 
\end{definition}

An $n$-canalizing function is also called a \emph{nested canalizing function} (NCF), and by definition all Boolean functions are $0$-canalizing.

Authors He and Macauley provided the following powerful stratification theorem. 

\begin{theorem}[\cite{he2016stratification}]\label{thm:he}
Every Boolean function $f(x_1,\ldots,x_n)\not\equiv 0$ can be uniquely written as 

$    f(x_1,\ldots,x_n) =$
\begin{equation}\label{eq:matts_theorem}
    M_1(M_2(\cdots (M_{r-1}(M_rp_C + 1) + 1)\cdots)+ 1)+ q,
\end{equation}

where each $M_i = \displaystyle\prod_{j=1}^{k_i} (x_{i_j} + a_{i_j})$ is a nonconstant extended monomial, $p_C$ is the \emph{core polynomial} of $f$, and $k = \displaystyle\sum_{i=1}^r k_i$ is the canalizing depth. Each $x_i$ appears in exactly one of $\{M_1,\ldots,M_r,p_C\}$, and the only restrictions are the following ``exceptional cases'':
\begin{enumerate}
    \item If $p_C\equiv 1$ and $r\neq 1$, then $k_r\geq 2$;
    \item If $p_C\equiv 1$ and $r = 1$ and $k_1=1$, then $q=0$.
\end{enumerate}
When $f$ is not canalizing (\textit{i.e.}, when $k=0$), we simply have $p_C = f$.
\end{theorem}



\begin{remark}~\label{rem:layer_output} Note the following properties of canalization.

(a) Theorem~\ref{thm:he} shows that any Boolean function has a unique extended monomial form, in which the variables are partitioned into different layers based on their dominance. Any variable that is canalizing (independent of the values of other variables) is in the first layer. Any variable that ``becomes'' canalizing when excluding all variables from the first layer is in the second layer, etc. Variables in any layer will be referred to as \emph{conditionally canalizing}. All remaining variables that never become canalizing are part of the core polynomial. The number of variables that eventually become canalizing, \textit{i.e.}, the number of conditionally canalizing variables, is the canalizing depth of the function. NCFs are exactly those functions where all variables are conditionally canalizing.

(b) While variables in the same layer may have different canalizing input values, they all share the same canalized output value, \textit{i.e.}, they all canalize a function to the same output. On the other hand, the outputs of two consecutive layers are distinct. Therefore, the number of layers of a $k$-canalizing function expressed as in Definition~\ref{def:kcanalizing} is simply one plus the number of changes in the vector of canalized outputs, $(b_1,b_2,\ldots,b_k)$.

\end{remark}



\begin{example}
\label{ex:theoretical}
The function $$f=(x_1 + 1)x_2\left[(x_3 + 1)x_4(x_5x_6 + x_7 + 1) + 1\right]$$ has the unique standard extended monomial form $f = M_1(M_2 p_C+ 1)$, where $M_1 = (x_1 + 1)x_2$, $M_2 = (x_3 + 1)x_4$, and $p_C=x_5x_6 + x_7 + 1$.  The variables $x_1$ and $x_2$ are canalizing and are in the first (outermost) layer, while $x_3$ and $x_4$ are conditionally canalizing and are in the second (and consequently last) layer.  In brief, $f$ has~2 layers and canalizing depth~4.    
\end{example}


We modify Equation~\ref{eq:matts_theorem} to present the \emph{layer structure}\footnote{In \cite{kadelka2017influence}, the authors referred to the summary vector $(k_1,\ldots,k_r)$, where $r$ is the number of layers and $k_i$ is the number of variables in the $i$-th layer, as the \emph{layer structure} of $f$.} of a Boolean function; that is complete information about canalizing and conditionally canalizing variables, their inputs and outputs, and the core function.

\begin{definition}
\label{def:layer_structure}
Let $f(x_1,\ldots,x_n)$ be a Boolean function with $r$ layers.   
The \emph{layer structure} of~$f$ is an ordered pair (\texttt{layers}, \texttt{core function}) where
\begin{itemize}
    \item \texttt{layers} is an ordered set of the layers (as a list) and corresponding canalized outputs (as an integer) of the form  
    \begin{center}\{(layer 1, output 1)$,\ldots,$ (layer $r$, output $r$)\}, \end{center}
    where each layer $i$, $1\leq i\leq r,$ is an unordered set of conditionally canalizing variables and their canalizing inputs of the form 
    \begin{center} \{(variable 1, input 1)$,\ldots,$ (variable $k_i$, input $k_i$)\}. 
    \end{center}
    \item \texttt{core function} is the function $f_C$ from Definition~\ref{def:kcanalizing}. The core function and core polynomial from Theorem \ref{thm:he} are intrinsically related, $f_C=p_C$ or $f_C=p_C+1$~\cite{he2016stratification}.
\end{itemize}
\end{definition}

\begin{example}
Consider the function $f$ in Example \ref{ex:theoretical}.  The layer structure of $f$ is 
%
%
\begin{align*}\Bigg( \bigg\{ \Big( \big\{ (x_1,1), (x_2,0) \big\} , 0 \Big),& \Big( \big\{ (x_3,1), (x_4,0) \big\} , 1 \Big) \bigg\}\ ,\\
&x_5x_6 + x_7 + 1 \Bigg).\end{align*}


\end{example}

\begin{remark}
\label{rem:unique}
Any nonzero Boolean function has a unique layer structure by Theorem~\ref{thm:he}.  Note that functions with a single essential variable will have two canalizing input and canalized output values, yet they still have a unique layer structure, forced by ``exceptional case''~2 in Theorem~\ref{thm:he}.  For example, the function $x_1$ has canalizing input value $0$ with corresponding canalized output value $0$ but also canalizing input value $1$ with corresponding canalized output value $1$.  However, its unique layer structure is $$\Bigg( \bigg\{ \Big( \big\{ (x_1,1) \big\} , 1 \Big) \bigg\}\ ,\ 1 \Bigg).$$  Likewise the unique layer structure for $f=x_1+1$ is  $$\Bigg( \bigg\{ \Big( \big\{ (x_1,0) \big\} , 1 \Big) \bigg\}\ ,\ 1 \Bigg)$$ since $f$ is written as $((x_1+a_1)\cdot p_C+1)+q$. Having one layer requires $p_C=1$.  By ``exceptional case''~2, we have $q=0$.  This now forces a canalizing input $a_1=0$ with corresponding canalized output of $1$.
\end{remark}



\section{Enumeration of Boolean functions with a given number of layers}
\label{sec:enumeration}
In this section, we count the number of Boolean functions on $n$ variables that contain a certain number of layers. This extends the results in~\cite[Section 5]{he2016stratification}. We start with a remark that summarizes the known results, which are all direct consequences of Theorem~\ref{thm:he}.

\begin{remark} \label{cor}The following has been proved in~\cite{he2016stratification}.\label{rem:he2016enumeration}

(a) The number $C_n$ of canalizing Boolean functions on $n\geq 0$ variables is 
$$C_n = 2((-1)^n-n-1) + \sum_{k=1}^n(-1)^{k+1}\binom nk 2^{k+1}2^{2^{n-k}}.$$

(b) The number $B^{\star}(n,0)$ of Boolean non-constant core polynomials on $n\geq 0$ variables (i.e., the number of Boolean non-constant, non-canalizing functions) is
$$B^{\star}(n,0) = 2^{2^n} - C_n - 2.$$

(c) The number $B(n,n)$ of Boolean nested canalizing functions on $n\geq 2$ variables is
$$B(n,n) = 2^{n+1} \sum_{r=1}^{n-1} \sum_{\substack{k_1+\cdots k_r = n\\ k_i\geq 1, k_r \geq 2}} \binom{n}{k_1,\ldots,k_r},$$
where $\binom{n}{k_1,\ldots,k_r} = \frac{n!}{k_1!k_2!\cdots k_r!}$. Further, $B(1,1) = 2$.


(d) The number $B(n,k)$ of Boolean functions on $n$ variables with exact canalizing depth $k$, for $1\leq k\leq n$, is
\begin{align*}B(n,k) = \binom nk &\Big[B(k,k) + B^\star(n-k,0)\cdot \\&2^{k+1}\sum_{r=1}^{k} \sum_{\substack{k_1+\cdots k_r = k\\k_i\geq 1}} \binom{n}{k_1,\ldots,k_r} \Big].\end{align*}


\end{remark}

\begin{corollary}\label{cor_enumeration}
For $n\geq 2$ and $1\leq r\leq n-1$, the number $B(n,n,r)$ of Boolean nested canalizing functions on $n$ variables with exactly $r$ layers is
$$B(n,n,r) = 2^{n+1} \sum_{\substack{k_1+\cdots k_r = n\\ k_i\geq 1, k_r \geq 2}} \binom{n}{k_1,\ldots,k_r},$$
and we have $B(n,n) = \sum_{r=1}^{n-1} B(n,n,r)$. Further, $B(1,1,1) = 2$.

The number $B(n,k,r)$ of Boolean functions on $n$ variables with exact canalizing depth $k$ and exactly $r$ layers, for $1\leq r\leq k\leq n$, is
\begin{align*}B(n,k,r) &= \binom nk \Big[B(k,k,r) + B^\star(n-k,0)\cdot\\
&\qquad 2^{k+1}\sum_{\substack{k_1+\cdots k_r = k\\k_i\geq 1}} \binom{n}{k_1,\ldots,k_r} \Big],\end{align*}
and we have $B(n,k) = \sum_{r=1}^{k} B(n,k,r)$. Note that if $n=k$ (i.e., for nested canalizing functions), we have $B^{\star}(n-k,0) = B^{\star}(0,0) = 0$ so that the provided formula does not contradict the $k_r\geq 2$ requirement.

Finally, the number of Boolean functions with exactly $r$ layers, for $1\leq r\leq n-1$, is $\sum_{k=1}^n B(n,k,r).$
\end{corollary}

\begin{proof}
To count the number of Boolean functions on $n$ variables that contain a fixed number of layers, we separate the summations in Remark~\ref{cor} (c) and (d) based on $r$, the number of layers.
\end{proof}

These formulas enable a simple exploration of the prevalence of canalizing functions. As known, the property of canalization, and in particular nested canalization, becomes increasingly rare for functions with more variables (Figure~\ref{new_figure}). Functions with $n-1$ canalizing variables all possess an nth non-essential variable and are, just like functions with $n-2$ canalizing variables, even rarer than nested canalizing functions (Figure~\ref{new_figure}A). When looking at the number of functions with a specific number of canalizing layers (instead of a specific number of canalizing variables), a monotonic relationship appears for $n\geq 3$: there exist fewer functions with $r$ layers than $r-1$ layers, $r=1,\ldots,n-1$.

\begin{figure*}[h]
    \label{new_figure}
    \centering
    \includegraphics[width=\textwidth]{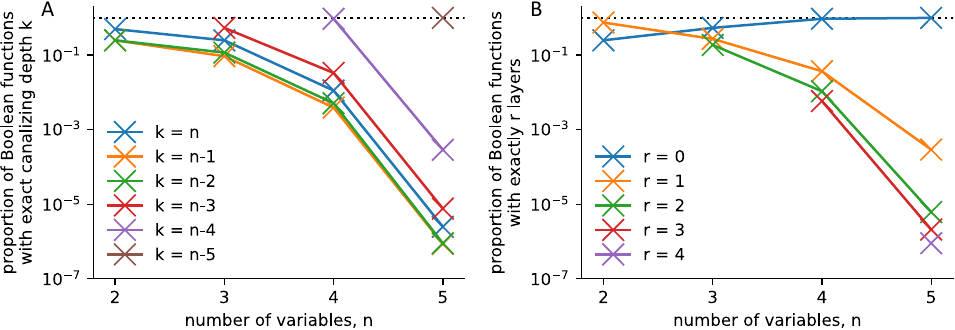}
    \caption{{\bf Enumeration of Boolean functions.} Using the formulas in Remark~\ref{cor} and Corollary~\ref{cor_enumeration}, this figure shows the proportion of Boolean functions with $n=2-5$ variables, which have (A) a specific canalizing depth or (B) a specific number of layers.}
    \label{fig:enumeration}
\end{figure*}

\section{Finding canalizing layers is \textbf{NP}-hard}
\label{sec:np-hard} 

In this section we show that the problem of determining the specific layer structure of a Boolean function is \textbf{NP}-hard, thus requiring efficient algorithms.
To show that the problem of determining the canalizing layers is \textbf{NP}-hard we exploit the connection between
the canalizing layers of a nested canalizing function and the function's symmetry groups given in~\cite{rosenkrantz2019symmetry}, which says that all variables in the same layer are contained within at most two symmetry groups. Here we reproduce the relevant result from~\cite{rosenkrantz2019symmetry} to make our description self-contained.

\begin{theorem}(adapted version of Theorem~2.5 in \cite{rosenkrantz2019symmetry})\label{thm:symmetry_nphard}
Unless $\mathbf{P} = \mathbf{NP}$, there is no polynomial time algorithm for determining the symmetry level of a Boolean function $f$ specified as a Boolean expression, even if $f$ is a CNF expression.
\end{theorem}

This theorem was proved by showing that determining the symmetry level of a Boolean function is equivalent to solving the CNF Satisfiability Problem (SAT), which is known to be \textbf{NP}-hard (Cook-Levin Theorem). We will now show that a polynomial time algorithm (with respect to the number of variables) for the identification of layers and core function of a Boolean function would also yield the symmetry level in polynomial time.

\begin{theorem}\label{thm:layers_nphard}
Unless $\mathbf{P} = \mathbf{NP}$, there is no polynomial time algorithm for determining the standard extended monomial form (Theorem~\ref{thm:he}; \textit{i.e.}, the layers and the core function) of a Boolean function~$f$ specified as a Boolean expression, even if $f$ is a CNF expression.
\end{theorem}
The following proof was adapted from  \cite{rosenkrantz2019symmetry}.
\begin{proof}[Proof of Theorem~\ref{thm:layers_nphard}]
Assume there is a polynomial time algorithm $\mathcal{A}$ (in the size of the inputs) that yields the canonical monomial form, as stated in Theorem~\ref{thm:he}, for any Boolean function.
We will show that $\mathcal{A}$ can be used to efficiently solve (in polynomial time) the CNF Satisfiability
problem (SAT) which is known to be \textbf{NP}-hard (Cook-Levin Theorem, see Section 2.6 of~\cite{garey1979}).

Let $g$ be a CNF formula representing an instance of SAT. Let $X = \{x_1, \dots, x_n\}$ denote the set of
Boolean variables used in $g$. Let $y_1 \notin \{x_1, \dots, x_{n}\}$
be a new variable and $h$ be the Boolean function with 1 layer: $h=y_1$.
Now, we create another CNF formula $f$ as $f = g\wedge h$. Note that $f$ is a function of $n+1$ variables, namely $x_1,\dots, x_n,y_1$.
Since $g$ is a CNF formula, so is~$f$. 

Consider the following cases for the satisfiability of $g$.
%
If $g$ is not satisfiable, then $f$ is also not satisfiable; that is, for all inputs, the value of $f$
is~0. By Definition~\ref{def:canalizing}, constant functions do not possess any canalizing variables. Thus the number of layers of $f$ is~0.
On the other hand, if $g$ is satisfiable, this implies that $y_1=0$ canalizes $f$ to $0$, \textit{i.e.}, $y_1$ is a canalizing variable of $f$. Therefore~$f$ has at least one layer.
%

%
%
%
%
%
%

Suppose we execute the polynomial time algorithm $\mathcal{A}$ on the function $f$ defined above. 
If $g$ is not satisfiable, then the number of layers of $f$ is 0.
Thus $\mathcal{A}$ will produce zero layers.
If $g$ is satisfiable, then the number of layers of $f$ is at least $1$; so, $\mathcal{A}$ will produce at least one
layer. In other words, $g$ is not satisfiable if and only if the number of layers produced by~$\mathcal{A}$ is~0.
Since $\mathcal{A}$ runs in polynomial time, we have an efficient algorithm for SAT, contradicting the assumption
that \textbf{P} $\neq$ \textbf{NP}.
\end{proof}


\section{Algorithms for revealing the canalizing layer structure}
\label{sec:algorithms}

In this section we present algorithms for finding the layers of canalization of a Boolean function in three different cases: 
\begin{itemize}
    \item[i.] The function is presented as a polynomial $f: \{0,1\}^n\rightarrow\{0,1\}$ (see Section~\ref{alg:1}).
    \item[ii.] The function is given in truth table format, \textit{i.e.}, as a vector (see Section~\ref{alg:2}).
    \item[iii.] The function is nested canalizing with a known order of conditionally canalizing variables (see Section~\ref{alg:3}). 
\end{itemize}
The output produced by Algorithms \ref{alg_layers1} and \ref{alg_layers2} follows the layer structure format in Definition~\ref{def:layer_structure}.


\subsection{Algorithm for Boolean functions represented as polynomials}\label{alg:1}

The first algorithm, implemented in Macaulay2, takes as input \emph{any} Boolean function and returns its canalizing layers. The algorithm is based on successive evaluations of the inputs of the 
function to check for canalizing inputs and subsequent removal of these from the search.

\begin{algorithm*}
\begin{algorithmic}[1]
\caption{Polynomial Form}
\label{alg_layers1}
\Require A Boolean function $f: \{0,1\}^n\rightarrow\{0,1\}$ as a polynomial.
\Ensure  Layer structure of $f$.

\If{$\#support(f)==0$} \Return {
$(\{\},f)$} 
\Comment{$f$ is constant}
\EndIf

\If{$\#support(f)==1$} \Return {
$(\{(\{(x_1,1)\},f(1))\},1)$ 
\Comment{Exceptional Cases:}

\Comment{$f=x_1$ or $f=x_1+1$}}
\EndIf

\State $Layers=\{\}$

\State Set $NewVar = support(f)$\Comment{get variables of $f$}
\While {$NewVar\neq \emptyset$}
\State $NewLayer=\{\}$ and $NegNewLayer=\{\}$
 \ForAll{ $x_i\in NewVar$ }
   \ForAll{$a \in \{0, 1 \}$}
    	\State $\hat{X} = (x_1,\dots,x_{i-1},x_i = a,x_{i+1},\dots,x_n)$\Comment{substitute $x_i=a$}
   	\State $g = f(\hat{X})$\Comment{evaluate at $x_i=a$}
   \EndFor 
   	\If {$support(g)==\emptyset$ }
		\State Append $(x_i,a)$ to $NewLayer$\Comment{store cond. canalizing variable and input}
		\State Append $(x_i,a+1)$ to $NegNewLayer$ \Comment{store noncanalizing input}
  	\EndIf           
 \EndFor
\State Append $(NewLayer, f(NewLayer))$ to $Layers$ 
\Comment{evaluate $f$ on canalizing inputs}

\Comment{of variables in $NewLayer$;} 

\Comment{append $NewLayer$ and}

\Comment{canalized output to $Layers$}
\State Set $NewVar=NewVar\backslash \{\text{variables in } NewLayer\}$
\State Set $f = f(NegNewLayer)$\Comment{evaluate $f$ on noncanalizing inputs}

\Comment{of variables in $NegNewLayer$} 
 \If {$NewLayer==\emptyset$}
    \State Break;
 \EndIf
\EndWhile \Comment{end of while loop}
\State \Return $Layers$, $f(NewVar)$\Comment{return layers and core function.}
\end{algorithmic}
\end{algorithm*}


\subsection{Algorithm for Boolean functions represented as truth tables}\label{alg:2}

The second algorithm, implemented in Python, uses elementary linear algebra and recursion to compute the layer structure of any Boolean function $f(x_1,\ldots,x_n)$ in $n$ variables.
Consider a \emph{binary truth table} represented as a $2^n \times (n+1)$-matrix over $\{0,1\}$ where the \emph{left-hand side} of dimension $2^n \times n$ corresponds to the inputs of a function $f:\{0,1\}^n\rightarrow \{0,1\}$ and the $(2^n\times 1)$-dimensional \emph{right-hand side} corresponds to the outputs of $f$ when evaluated on the inputs. Like most authors (e.g., \cite{he2016stratification}), we exclude constant functions from the set of canalizing functions (Definition~\ref{def:canalizing}). In Algorithm~\ref{alg_layers2}, we therefore first check if $f$ is constant (row 3). In the following, assume $f$ is not constant.

Let $\mathbf T_n = (\mathbf t_1,\ldots,\mathbf t_n) \in \{0,1\}^{2^n \times n}$ be the left-hand side of a binary truth table. That is, $\mathbf t_i$ is a binary vector of length $2^n$, which is $1$ in the rows of $\mathbf T_n$ where $x_i=1$, and $0$ otherwise. Likewise, we will think of the right-hand side of a binary truth table as a vector $f$.
With $\langle \cdot, \cdot \rangle$ denoting the dot product, we have the following for any non-constant Boolean function~$f$:
\begin{itemize}
    \item $x_i=1$ canalizes $f$ to the value $1$ if and only if $\langle \mathbf t_i,f \rangle = 2^{n-1}$;
    \item $x_i=0$ canalizes $f$ to the value $1$ if and only if $\langle \mathbf 1-\mathbf t_i,f \rangle = 2^{n-1}$;
    \item $x_i=1$ canalizes $f$ to the value $0$ if and only if $\langle \mathbf t_i,\mathbf 1-f \rangle = 2^{n-1}$;
    \item $x_i=0$ canalizes $f$ to the value $0$ if and only if $\langle \mathbf 1-\mathbf t_i,\mathbf 1-f \rangle = 2^{n-1}$.

\end{itemize}




With this we can define four sets
\begin{align*}
I_{1\to 1}(f) &:= \Big\{i\in\{1,\ldots,n\}~\Big|~ \langle \mathbf t_i,f \rangle = 2^{n-1} \Big\},\\
I_{0\to 1}(f) &:= \Big\{i\in\{1,\ldots,n\}~\Big|~  \langle  \mathbf 1-\mathbf t_i,f \rangle = 2^{n-1} \Big\},\\
I_{1\to 0}(f) &:= \Big\{i\in\{1,\ldots,n\}~\Big|~ \langle \mathbf t_i,\mathbf 1-f \rangle = 2^{n-1} \Big\},\\
I_{0\to 0}(f) &:= \Big\{i\in\{1,\ldots,n\}~\Big|~ \langle  \mathbf1-\mathbf t_i,\mathbf 1-f \rangle = 2^{n-1} \Big\}.
\end{align*}
Then $I_{a\to b}(f)$ contains all the indices of canalizing variables that canalize $f$ to the canalized output value $b$ if they take on the canalizing input value $a$, for $a,b\in\{0,1\}$. For brevity, let $I_{b}(f) = I_{0\to b}(f) \cup I_{1\to b}(f)$. In Algorithm~\ref{alg_layers2}, we do not calculate these sets for constant functions, which are by Definition~\ref{def:canalizing} not canalizing. Therefore we have for any non-constant~$f$ 
\begin{itemize}
    \item $|I_0(f) \cup I_1(f)|$ is the number of canalizing variables.
    \item $I_0(f) = I_1(f) = \emptyset \Longleftrightarrow f$ is not canalizing.
\end{itemize}
By Remark~\ref{rem:layer_output}, all canalizing variables canalize $f$ to the same output value and by Remark~\ref{rem:unique} each canalizing variable in functions with more than one essential variable possesses a unique canalized output. Thus, for non-trivial functions -  with more than one essential variable - we have 
\begin{itemize}
    \item $I_0(f) \neq \emptyset \Longrightarrow I_1(f) = \emptyset$;
    \item $I_1(f) \neq \emptyset \Longrightarrow I_0(f) = \emptyset$.
\end{itemize}

This allows us to use these sets to define an iterative process that finds all layers of a Boolean function, recording layer by layer the canalizing inputs, canalized outputs and conditionally canalizing variables (see Algorithm~\ref{alg_layers2}). Note that Boolean functions with a single essential variable have two choices for the canalized output value (Remark~\ref{rem:unique}). Theorem~\ref{thm:he} (``exceptional case''~2) ensures a unique layer structure for these functions by forcing a canalized output of~$1$. To agree with this default choice in Algorithm~\ref{alg_layers2}, we first check for canalizing variables in the set $I_1(f)$.

\begin{example}
Let $n=3$ and consider the left-hand side of the $(2^3\times 3)$-truth table $$\mathbf T_3 = \begin{pmatrix}
0 & 0 & 0\\
0 & 0 & 1\\
0 & 1 & 0\\
0 & 1 & 1\\
1 & 0 & 0\\
1 & 0 & 1\\
1 & 1 & 0\\
1 & 1 & 1
\end{pmatrix}.$$
%
Let $f(x_1,x_2,x_3)=x_1\land (x_2 \lor x_3)$. As a  vector (with variable order $x_1,x_2,x_3$), we can express~$f$ as $f = (0,0,0,0,0,1,1,1)$. 
Then we have 
$$f\cdot \mathbf T_3 = (3,2,2),\ f\cdot (\mathbf 1-\mathbf T_3) = (0,1,1),$$ 
where $\mathbf 1-\mathbf T_3$ is the difference of the $(8\times 3)$-matrix of ones and $\mathbf T_3$ modulo 2.
None of these entries equal $2^{3-1}=4$; hence $I_1(f) = \emptyset$.

On the other hand, $$(1-f)\cdot \mathbf T_3 = (1,2,2),\ (1-f)\cdot (\mathbf 1-\mathbf T_3) = (4,3,3).$$ Thus $I_0(f) = \{1\}$ and $x_1$ is therefore the only canalizing variable of $f$ with $x_1 = 0$ canalizing $f$ to $0$. Removing the first layer $\{x_1\}$ from $f$ yields the subfunction $g = f(1,x_2,x_3) = x_2 \lor x_3$ or $(0,1,1,1)$ as a  vector. 

Now we repeat the process and find that $I_0(g) = \emptyset$ since $(1-g)\cdot \mathbf T_2=(0,0)$ and $(1-g)\cdot (\mathbf 1-\mathbf T_2)=(1,1)$.  Likewise $I_1(g) = \{2,3\}$ since $g\cdot \mathbf T_2=(2,2)$ and $g\cdot (\mathbf 1-\mathbf T_2)=(1,1)$. Therefore $x_2$ and $x_3$ are canalizing variables of $g$ and thus conditionally canalizing variables of $f$. Thus $\{x_2,x_3\}$ forms the second layer of $f$ with canalized output~1. The process is completed because the new subfunction $h=g(1,1)=1$ contains no more variables. Note that for non-NCFs, the process completes when no new conditionally canalizing variables are found, \textit{i.e.}, when both $I_0$ and $I_1$ are empty. 
\end{example}

\begin{algorithm*}
\begin{algorithmic}[1]
\caption{Truth Table Form}
\label{alg_layers2}
\Require A Boolean function $f: \{0,1\}^n\rightarrow\{0,1\}$ as a vector of length $2^n$.
\Ensure  Layer structure of $f$.
\State $Layers=\{\}$\Comment{stores all layers}\label{repeat}
\State Set $n=\log_2($length\ $ f)$ \Comment{ number of variables of $f$; $n$ and $f$ will change} 
\If {$support(f)==\emptyset$ } \Comment{$f$ is constant}
    \State \Return $(Layers,f)$ \Comment{constant functions are not canalizing - Def~\ref{def:canalizing}}
\EndIf
\State Initialize $\mathbf T_n$ \Comment{left-hand side of a truth table} 
\State Set $v=(1,\ldots ,1)$ \Comment{vector of $2^n$ ones} 
\State Set $M=[v',\ldots ,v']$ \Comment{matrix of $2^n\times n$ ones; $v'$ = transpose of $v$}
\State  $NewLayer=\{\}$ \Comment{stores the current layer}
\State  $NegNewLayer=\{\}$ \Comment{stores noncanalizing input of current layer}
\State Compute $(v-f)\cdot \mathbf T_n$ and $(v-f)\cdot (M-\mathbf T_n)$
\State Derive $I_{0\to 1}(f)$, $I_{1\to 1}(f)$, and $I_1(f)$
\If{$I_1(f)\neq \emptyset$}
    \ForAll{$a\in\{0,1\}$}
        \ForAll{$x_i\in I_{a\to 1}$}
            \State Append $(x_i,a)$ to $NewLayer$ \Comment{store cond. canalizing variable and  input}
		    \State Append $(x_i,a+1)$ to $NegNewLayer$ \Comment{store noncanalizing input}
        \EndFor
    \EndFor
        \State Append ($NewLayer$,1) to $Layers$\Comment{store canalized output}
\Else{
    \State Compute $f\cdot \mathbf T_n$ and $f\cdot (M-\mathbf T_n)$
    \State Derive $I_{0\to 0}(f)$, $I_{1\to 0}(f)$, and $I_0(f)$
    \If{$I_0(f)\neq \emptyset$}
        \ForAll{$a\in\{0,1\}$}
        \ForAll{$x_i\in I_{a\to 0}$}
            \State Append $(x_i,a)$ to $NewLayer$ \Comment{store cond. canalizing variable and input}
		    \State Append $(x_i,a+1)$ to $NegNewLayer$ \Comment{store noncanalizing input}
        \EndFor
    \EndFor
        \State Append ($NewLayer$,0) to $Layers$\Comment{store canalized output}
    \Else~\Return $(Layers,f)$ \Comment{return layers and core function}
    \EndIf
}
\EndIf
\State Set $f = f( NegNewLayer)$ 
\Comment{evaluate $f$ on noncanalizing inputs}

\Comment{of variables in $NegNewLayer$;} 
\State \textbf{Goto} Step \ref{repeat} \Comment{repeat using the new $f$}
\end{algorithmic}
\end{algorithm*}


\subsection{Algorithm for nested canalizing functions}
\label{alg:3}

The authors of \cite{jarrah2007nested} developed a parametrization of NCFs corresponding to points in the affine space $\mathbb{F}_2^{2^n}$ that satisfy a certain collection of polynomial equations. They observed that since the terms of a Boolean polynomial consist of square-free monomials, one can uniquely index monomials by the subsets of $[n]=\{1,\ldots,n\}$ corresponding to the variables appearing in the monomial. Thus the set of all Boolean polynomials can be expressed as 
\begin{equation}\label{coeffs}
\displaystyle \Bigg\{\sum_{S\subseteq [n]}c_s\prod_{i\in S}x_i \,\Bigg|\, c_s\in \mathbb F_2\Bigg\}.
\end{equation}
We use this parametrization, and in particular Corollary~3.6 
in~\cite{jarrah2007nested}, 
to suggest a linear time algorithm for finding the canalizing layers of a given NCF. The algorithm is based on the following proposition, which is a straightforward consequence of Corollary 3.6 in~\cite{jarrah2007nested}.

\begin{proposition}

Let $f$ be a Boolean polynomial written in the form of (\ref{coeffs}).
%
If for all $j=1,\ldots, n-1$
$$
c_{\{1,\ldots, j\}}=c_{\{1,\ldots, j,j+1\}}c_{[n]\setminus \{1,\ldots, j,j+1\}},
$$
then $x_j$ and $x_{j+1}$ are in the same layer.
\end{proposition}\label{thm:lin-time}

Therefore if the canalizing order of the $n$ variables of an NCF is known, one needs to check at most $n-1$ equalities; that is, the algorithm suggested by Proposition~\ref{thm:lin-time} will take linear time in the number of variables.

\subsection{Time complexity of the algorithms}
\label{time}

As an indication for the time complexity, we recorded the average run time of both algorithms for random Boolean functions in $4$ to $16$ variables. We tested two extremes: the run time for random noncanalizing functions and for random nested canalizing (\textit{i.e.}, most canalizing) functions. For each even number of variables $n\in [4,16]$ and for both types of functions, we randomly generated five Boolean functions of that type as a  vector of length $2^n$, which served as the input for Algorithm~\ref{alg_layers2}. To use the same sets of random functions in Algorithm~\ref{alg_layers1}, we transformed the  vectors into polynomials in disjunctive normal form.

At small numbers of variables ($n\leq 10)$, both algorithms were on average faster at determining that there were no canalizing variables than at finding all layers of an NCF~(Figure~\ref{fig:time}A). Interestingly for functions in many variables, this switched for Algorithm~\ref{alg_layers1}; its run time increased only marginally for NCFs, much slower than for noncanalizing functions. This is likely because Algorithm~\ref{alg_layers1} benefits from a high amount of structure in nested canalizing polynomials, highlighting the advantage of the use of the polynomial representation over truth tables (\textit{i.e.}, vectors). Note that the algorithms are implemented in two different programming languages, which prohibits a direct comparison of the run time. 

Algorithm~\ref{alg_layers2} can be sped up by pre-computing the left-hand side of the truth table $\mathbf T_m$ for all $m$ with $1\leq m \leq n$. In our limited experiment this reduced the average run time by $38.6\%$ ($n=4$) to $95.0\%$ ($n=20$) for noncanalizing functions and by $12.8\%$ ($n=4)$ to $83.4\%$ ($n=20$) for NCFs~(Figure~\ref{fig:time}B).

\begin{figure*}[h]
    \centering
    \includegraphics[width=\textwidth]{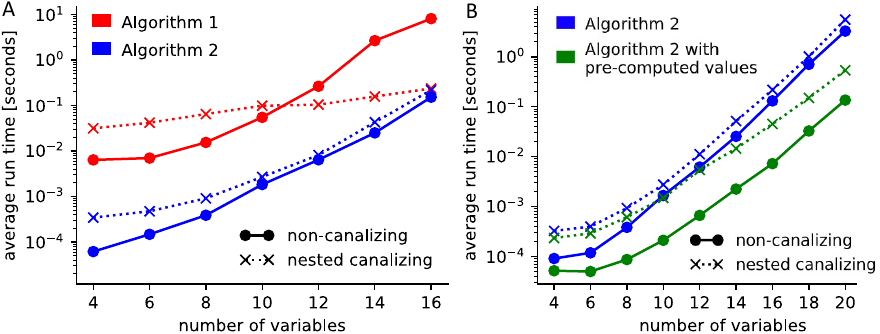}
    \caption{{\bf Comparison of the run time of the algorithms.} (A) For both algorithms, we compared the average computation time of the layers for noncanalizing (solid lines with circles) and nested canalizing (dotted lines with {\sffamily x}) Boolean functions with a fixed number of variables, $n$ ($x$-axis). (B) Comparison of the basic version of Algorithm 2 (blue) and a modified version of the algorithm where the left-hand side of the truth table $\mathbf T_m$ has been pre-computed for all~$m$ with $1\leq m \leq n$.}
    \label{fig:time}
\end{figure*}

The findings from these bench-marking experiments align with theoretical computations of the worst-case time complexity, which is $O(n^22^n)$ for both algorithms. For Algorithm~\ref{alg_layers1}, the number of checks (in the worst case where, in each iteration, the last evaluated variable is the only canalizing) is
\[
n+n-1+\cdots+2+0=\frac{n(n+1)}{2}-1 = O(n^2)
\]
(notice the zero because the last layer of an NCF contains at least two variables, Theorem~\ref{thm:he}) and each of these checks requires $O(2^n)$  polynomial evaluations. Thus, Algorithm~\ref{alg_layers1} is $O(n^22^n)$. For Algorithm~\ref{alg_layers2}, the matrices that are evaluated to find canalizing variables are of size $n2^n$ (Lines 11 and 20 in Algorithm~\ref{alg_layers2}) and in the worst case, we find $n$ canalizing variables at once yielding another iteration over $n$ variables (Lines 14 and 23 Algorithm~\ref{alg_layers2}). Thus, Algorithm~\ref{alg_layers2} is $O(n^22^n)$.

\section{Applications}
\label{sec:applications}
In this section, we present two applications that utilize the canalizing structure of an NCF: converting the function into disjunctive normal form (Section~\ref{sec:DNF}), and reverse-engineering an unknown function from (partial) knowledge of its layers (Section~\ref{sec:rev-eng}).
\subsection{Finding disjunctive normal forms using canalizing layers}
\label{sec:DNF}


A disjunctive normal form (DNF) is a canonical representation of Boolean functions, which consists of a disjunction of conjunctions. A disjunctive normal form is a useful representation of Boolean functions that has been used for many applications.
For instance, a minimal DNF is needed for the identification of stable motifs~\cite{Zanudo:2015aa,yang18}.
In this section we will use the canalizing layers  to obtain a disjunctive normal form expression for a nested canalizing function. 
In general, the Quine-McCluskey algorithm~\cite{Quine1952} can be used to find a disjunctive normal form for any Boolean function.

The following proposition is a special case of Theorem~\ref{thm:he} where we assume that the Boolean function 
is nested canalizing.

\begin{proposition}\label{prop:ncf_layers}
Every nested canalizing function $f(x_1,\dots,x_n)$ can be uniquely written as
\begin{equation}
\label{eq:layers}
f(x_1,\dots,x_n) = M_1(M_2(\dots(M_{r-1}(M_r+1)+1)\dots)+1)+b,
\end{equation}
where $M_i = \displaystyle\prod_{j=1}^{k_i}(y_{i_j}+1)$, $y_{i_j}\in\{x_{i_j},\overline{x}_{i_j}\}$, and $n = k_1+\cdots+ k_r$.
Each variable $y_i$ appears in exactly one of the $M_1,\dots,M_r$. 
\end{proposition}

\begin{remark}\label{rem:ncf_layers}~
\begin{enumerate}
  \item Note that the extended monomials in Proposition~\ref{prop:ncf_layers} can be written as
\begin{displaymath}
M_i = \prod_{j=1}^{k_i}(y_{i_j}+1) = \overline{y}_{i_1}\wedge\cdots\wedge \overline{y}_{i_{k_i}}.
\end{displaymath}
  \item The negated form of $f$ is $f+1$. Thus to obtain the canalizing layers description of $\overline{f}$ we add~1 to Equation~\ref{eq:layers}:
\begin{equation}
\label{eq:layers1}
\begin{array}{l}
\overline{f}(x_1,\ldots,x_n) = \\
M_1(M_2(\cdots(M_{r-1}(M_r+1)+1)\cdots)+1)+b+1.
\end{array}
\end{equation}  
\end{enumerate}
\end{remark}

Now using the first part of Remark~\ref{rem:ncf_layers}, we introduce the following notation
\begin{displaymath}
\overline{M}_i = y_{i_1}\vee\cdots\vee y_{i_{k_i}}.
\end{displaymath}

Given a nested canalizing function $f$ written in the format of Equation~\ref{eq:layers}, one can obtain
a disjunctive normal form for~$f$. 

\begin{theorem}
Let $f$ be a nested canalizing function that is written as in Equation~\ref{eq:layers}.
Then the following formulas provide a disjunctive normal form for $f$ depending of the value of $b$ in Equation~\ref{eq:layers}.
\begin{enumerate}\label{thm:dnf}
    \item If $b=0$, then a DNF of $f$ is given by
    \begin{equation}
\label{eq:DNF0}
f = \bigvee_{i=0}^{\lfloor r/2\rfloor}\left(M_{2i+1}\wedge\bigwedge_{j=1}^i\overline{M}_{2j} \right),
\end{equation}
where $M_{r+1}=1$ if needed.
    \item If $b=1$, then a DNF of $f$ is given by
    \begin{equation}
\label{eq:DNF1}
f = \bigvee_{i=0}^{\lfloor r/2\rfloor}\left(\overline{M}_{2i+1}\wedge\bigwedge_{j=1}^i M_{2j} \right),
\end{equation}
where $\overline{M}_{r+1}=1$ if needed.
\end{enumerate}
\end{theorem}
\begin{proof} We consider two cases for the value of $b$.
\begin{enumerate}
    \item If $b=0$, then from Equation~\ref{eq:layers} $f$ can be written as
\[
\begin{array}{l}
f(x_1,\dots,x_n) = \\
M_1(M_2(\dots(M_{r-1}(M_r+1)+1)\dots)+1),
\end{array}
\] 
which can also be written as
\[
f = \bigvee_{i=0}^{\lfloor r/2\rfloor}\left(M_{2i+1}\wedge\bigwedge_{j=1}^i\overline{M}_{2j} \right).
\]
Then notice that
\[
\begin{array}{ll}
M_{2i+1}\wedge\overline{M}_{2j} =& M_{2i+1}\wedge(y_{i_1}\vee\cdots\vee y_{i_{k_i}}) =   \\
&\displaystyle\bigvee_{k=1}^{k_i}M_{2i+1}\wedge y_{i_j}.
\end{array}
\]
Thus, Equation~\ref{eq:DNF0} is in disjunctive normal form.
\item If $b=1$, then from Equation~\ref{eq:layers} $f$ can be written as
\[
\begin{array}{l}
f(x_1,\dots,x_n) = \\
M_1(M_2(\dots(M_{r-1}(M_r+1)+1)\dots)+1)+1,
\end{array}
\] 
which can also be written as    
\[    
f = \bigvee_{i=0}^{\lfloor r/2\rfloor}\left(\overline{M}_{2i+1}\wedge\bigwedge_{j=1}^i M_{2j} \right),
\]
Then notice that
\[
 \overline{M}_{2i+1}\wedge M_{2j} = (y_{i_1}\vee\cdots\vee y_{i_{k_i}})\wedge M_{2j} = \bigvee_{k=1}^{k_i}y_{i_j}\wedge M_{2j}.
\]
Thus Equation~\ref{eq:DNF1} is in disjunctive normal form.
\end{enumerate}

\end{proof}

\begin{remark} For Theorem~\ref{thm:dnf}, when $i=0$, the term inside the parentheses of Equation~\ref{eq:DNF0} (Equation~\ref{eq:DNF1}) reduces to $M_1$ ($\overline{M}_{1}$).
\end{remark}

\begin{example}
Let $f = x_1x_2x_3x_4 + x_1x_2x_3 + x_1x_2x_4 + x_2x_3x_4 + x_1x_2 + x_1x_3 + x_2x_3 + x_2x_4 + x_1 + x_2 + x_3 + 1$.
Using Algorithm~\ref{alg:1}, we find that the layer structure of this function is $f = M_1(M_2+1)$ where $M_1=(x_1+1)(x_3+1)$ and $M_2=x_2(x_4+1)$.

Thus the value of $b$ in Equation~\ref{eq:layers} is $b=0$. Then from Equation~\ref{eq:DNF0}, the disjunctive normal form is $f = M_1\wedge \overline{M}_2 = (M_1 \wedge \overline{x}_2)\vee (M_1 \wedge x_4)=
(\overline{x}_1\wedge \overline{x}_3 \wedge \overline{x}_2)\vee (\overline{x}_1\wedge \overline{x}_3 \wedge x_4)$.
\end{example}

\begin{example}
Let $f = x_1x_2x_3x_4 + x_1x_2x_3 + x_2x_3x_4 + x_2x_3 + x_4$. Using Algorithm~\ref{alg:1}, we find that
the layer structure of this function is $f = M_1(M_2+1)+1$ where $M_1=x_4+1$ and $M_2=(x_1+1)x_2x_3$. 
Thus the value of $b$ in Equation~\ref{eq:layers} is $b=1$.
Then from Equation~\ref{eq:DNF1}, the disjunctive normal form is $f = \overline{M}_1\vee M_2 = x_4\vee (\overline{x}_1\wedge x_2 \wedge x_3)$.
\end{example}


\subsection{Reverse engineering nested canalizing functions}
\label{sec:rev-eng}

When part of the canalizing layer structure of an NCF is known, one can generate all NCFs with that structure using the parametrization of NCFs given in (\ref{coeffs}), together with the following corollary to Theorem~\ref{thm:he}.

\begin{corollary}\label{cor:factor}

Let $f(x_1,\ldots,x_n)$ be a polynomial over $\mathbb{F}_2$. 

\begin{enumerate}

\item A variable $x_i$ is canalizing in the first layer of $f$ with canalizing input 0 if and only if
$f(x_1,\ldots,x_n)-f(0,\dots,0)=x_i\cdot g$ for some function $g$.

\item A variable $x_i$ is canalizing in the first layer of $f$ with canalizing input 1 if and only if 
$f(x_1,\ldots,x_i+1,\ldots,x_n)-f(0,\ldots,0, x_i=1,0,\ldots,0)=x_i\cdot g$ for some function $g$.

\end{enumerate}

\end{corollary}

\begin{proof}
Part 1 follows directly from Theorem~\ref{thm:he}. Part 2 also follows from Theorem~\ref{thm:he} by noticing that the constant term of $f(x_1,\ldots,x_i+1,\ldots,x_n)$ is obtained by setting $x_i=1$ and $x_j=0$ for $j\ne i$, so it is given by 
$f(0,\ldots, 0,x_i=1,0,\ldots,0)$.
\end{proof}

\begin{example}
Let $f=x_1x_2+x_1+1\in \mathbb{F}_2[x_1,x_2]$. Then $f(x_1,x_2)-f(0,0)=x_1x_2+x_1+1-1=x_1x_2+x_1$ and since it can be factored as $x_1(x_2+1)$, we see that $x_1$ is a 
 canalizing variable of $f$ with canalizing input 0. Similarly, $f(x_1,x_2+1)-f(0,1)=x_1x_2$, so we have that $x_2$ is another canalizing variable of $f$ with canalizing input 1.
\end{example}


Corollary~\ref{cor:factor} allows for quick identification of the canalizing variables in the outermost layer of a function and we will use it as we ``reverse engineer'' an NCF from partial information about its canalizing structure. To illustrate the process we will consider the regulation of the mammalian cell cycle and specifically how the members of the E2F family of transcription factors are regulated by other elements. This example is based on the work in~\cite{faure06} in which the mammalian cell cycle is modeled using a logical framework.


Mammalian cell division is tightly controlled and coordinated with the overall growth of the organism. The protein Rb forms a complex with members of the E2F family of transcription factors, turning them from transcriptional activators to repressors. As a result, Rb is considered an inhibitor of E2F and so is CycB. CycA plays a more complicated role in the regulation of E2F. While generally an inhibitor of E2F, CycA's effect on E2F is known to be altered by p27 but we will assume we do not know how. 

Thus the layer structure of the polynomial that models the dynamics of E2F will be an NCF in variables $x_1=$ CycB, $x_2=$ Rb, $x_3=$ p27, and $x_4=$ CycA distributed across two layers as follows: \emph{Layer 1} has canalizing output 0 and contains $x_1$ with canalizing input 1 and $x_2$ also with canalizing input 1;
\emph{Layer 2} has unknown canalizing output and contains $x_3$ with unknown canalizing input and $x_4$ with canalizing input 0. The goal is to find \emph{all} NCFs with this prescribed canalizing structure.


First, we will use Theorem 3.3 in~\cite{jarrah2007nested} which provides relationships among the coefficients in the NCF. (For ease of reading, we write below, for example, $c_{123}$ instead of $c_{\{1,2,3\}}$.) Notice that not all the equations that follow from the theorem carry useful information. Here we present only those that are not trivially true. 
\begin{eqnarray}
c_{1234} &=& 1\label{eq:first}\\
c_2&=&c_{12}c_{234}\\
c_3&=&c_{123}c_{234}c_{134}\\
c_4&=&c_{234}c_{134}c_{124}\\
c_{13}&=&c_{123}c_{134}\\
c_{14}&=&c_{134}c_{124}\\
c_{23}&=&c_{123}c_{234}\\
c_{24}&=&c_{234}c_{124}\\
c_{34}&=&c_{234}c_{134}
\end{eqnarray}
Second, we use Corollary 3.6 in~\cite{jarrah2007nested} and the fact that variables in the same layer have the same canalizing output, while those in adjacent layers have different canalizing output to ensure that $x_1$ and $x_2$ are in one layer and $x_3$ and $x_4$ and in a different layer. The corollary produces the following nontrivial equations:

\begin{eqnarray}
c_1 &=& c_{12}c_{134}\\
c_{12}&=&c_{123}c_{124}+1\label{eq:last}
\end{eqnarray}
We solve the system of eleven equations 
(\ref{eq:first})-(\ref{eq:last}) using Gr\"obner bases with lexicographic order to find that the system has five free parameters, $c_{134}=a,c_{234}=b,c_{124}=c,c_{123}=d$, and $c_{\emptyset}=e$. The rest of the coefficients can be expressed in terms of these free parameters as $c_1=acd+a, c_2=bcd+b, c_3=abd, c_4=abc, c_{12}=cd+1, c_{13}=ad, c_{14}=ac, c_{23}=bd, c_{24}=bc$, and $c_{34}=ab$. Thus the family of functions that satisfy the prescribed canalizing structure is of the form
\[
\begin{array}{l}
f=x_1x_2x_3x_4 + (acd + a)x_1 + (bcd + b)x_2 + abdx_3 \\
+ abcx_4 + (cd + 1)x_1x_2 + adx_1x_3 + acx_1x_4 + bdx_2x_3 + \\
 bcx_2x_4 + abx_3x_4 + ax_1x_3x_4 + bx_2x_3x_4 + cx_1x_2x_4 + \\
 dx_1x_2x_3 + e.
\end{array}
\]
So far, we have ensured that the above family of functions are nested canalizing in canalizing order $x_1,x_2,x_3,x_4$ and that $x_1$ and $x_2$ are in one layer and $x_3$ and $x_4$ are in a different layer. We have not guaranteed $x_1$ and $x_2$ will be in the first layer or that the variables will have the required canalizing input and corresponding output. 

To ensure that variables $x_1$ and $x_2$ will be in the first layer and both have canalizing input~1, we will use Corollary~\ref{cor:factor} and consider $f(x_1+1,x_2+1,x_3,x_4)-f(1,1,0,0)$, which expanded is
$$
x_1 + a x_1 + c d x_1 + a c d x_1 + x_2 + b x_2 + c d x_2 + b c d x_2 + 
$$
$$
 x_1 x_2 + c d x_1 x_2 + d x_3 + a d x_3 + b d x_3 + a b d x_3 + d x_1 x_3 + 
 $$
 $$
 a d x_1 x_3 + d x_2 x_3 + b d x_2 x_3 + d x_1 x_2 x_3 + c x_4 + a c x_4 + 
 $$
$$
 b c x_4 + a b c x_4 + c x_1 x_4 + a c x_1 x_4 + c x_2 x_4 + b c x_2 x_4 + 
 $$
 $$
 c x_1 x_2 x_4 + x_3 x_4 + a x_3 x_4 + b x_3 x_4 + a b x_3 x_4 + x_1 x_3 x_4 + 
 $$
$$
 a x_1 x_3 x_4 + x_2 x_3 x_4 + b x_2 x_3 x_4 + x_1 x_2 x_3 x_4.
$$
In order to be able to factor out $x_1$ and $x_2$ from all nonconstant terms, we need to select zero coefficients for all terms that do not have both $x_1$ and $x_2$, that is, set to zero the coefficients for the terms $x_1, x_2, x_3, x_4, x_1x_3, x_2x_3$, etc. This generates a system of 11 nonlinear equations in variables $a,b,c,d$ over $\mathbb{F}_2$ which has (lexicographic) Gr\"obner basis $\{a+1, b+1\}$. So we set $a=b=1$ in~$f$ to obtain
$$
f'=e + x_1 + c d x_1 + x_2 + c d x_2 + x_1 x_2 + c d x_1 x_2 + 
$$
$$d x_3 + d x_1 x_3 + 
 d x_2 x_3 +  d x_1 x_2 x_3 +
 $$
$$
c x_4 +c x_1 x_4 +c x_2 x_4 +c x_1 x_2 x_4 + x_3$$
$$
x_4 + x_1 x_3 x_4 + x_2 x_3 x_4 + x_1 x_2 x_3 x_4.
$$
Now notice that when we plug $x_1$'s or $x_2$'s canalizing value into $f'$, we get $f'(1,x_2,x_3,x_4)=f'(x_1,1,x_3,x_4)=1+cd+e$. Since we want the output to be~0, we set $e=cd+1$. Furthermore, since we have already guaranteed that $x_1$ and $x_2$ will be in the first layer with canalizing input 1, we can now set them equal to their noncanalizing input, 0, and proceed with the next layer, working with the much simpler function $f'(0,0,x_3,x_4)=1 + c d + d x_3 + c x_4 + x_3 x_4$.

We now want to make sure that $x_4$ has canalizing input~0. (Recall that we made sure that $x_3$ and $x_4$ are in a different layer by using Corollary 3.6 in~\cite{jarrah2007nested}.) Recalling Corollary~\ref{cor:factor}, we need to be able to factor out $x_4$ from all nonconstant terms of $f'(0,0,x_3,x_4)$ which requires us to set $d=0$. The resulting function is $1 + c x_4 + x_3 x_4$. 
We can quickly check that, indeed, when $x_4$ takes on its canalizing input of~0, the output is a constant, namely~1. Notice that although we did not explicitly prescribe the canalizing output of $x_4$, we made sure it is different from the canalizing output of Layer 1 or else $x_4$ would not be in a different layer from $x_1$ and $x_2$. The function $1 + c x_4 + x_3 x_4$ depends on a single parameter $c$ whose value decides the canalizing input of $x_3$ (which we kept undetermined). When $x_3=0$, the expression $1 + c x_4 + x_3 x_4$ evaluates to $1+cx_4$, meaning that if the canalizing input of $x_3$ is 0, then we must set $c=0$; if $x_3=1$, the function evaluates to $1+(c+1)x_4$, \textit{i.e.}, if the canalizing input of $x_3$ is 1, then we must set $c=1$. Since the canalizing input of $x_3$ is to be left undetermined, we will keep $c$ as a parameter and conclude that the one-parameter family of nested canalizing functions with the required canalizing structure is the set
\[
\begin{array}{l}
\{1 + x_1 + x_2 + x_1 x_2 + c x_4 + c x_1 x_4 + c x_2 x_4 + c x_1 x_2 x_4 + \\
x_3 x_4 + x_1 x_3 x_4 + x_2 x_3 x_4 + x_1 x_2 x_3 x_4 ~\vert~ c\in\mathbb{F}_2\}.
\end{array}
\]
This family of NCFs is consistent with the model proposed in~\cite{faure06}. With the same variable naming convention, their Boolean equation for E2F is 
$
(\overline{x_2}\land \overline{x_4}\land \overline{x_1})\lor(x_3\land \overline{x_2}\land \overline{x_1}),
$
which written as a polynomial over $\mathbb F_2$ is 
$$
\begin{array}{l}
1 + x_1 + x_2 + x_1 x_2 + x_4 + x_1 x_4 + x_2 x_4 + x_1 x_2 x_4 +  \\
x_3 x_4 + x_1 x_3 x_4 + x_2 x_3 x_4 + x_1 x_2 x_3 x_4,
\end{array}
$$
matching the form of our family of NCFs when $c=1$. Recall that for the sake of our example we chose to assume that we did not know the role of p27 in altering the regulatory effect of CycA on E2F. Suppose that we now learn that, when present, p27 does not allow the inhibition of E2F by CycA. Equipped with this information we set $c=1$ so that when both CycA and p27 are present, E2F remains present, arriving to the same model for E2F as the one proposed in~\cite{faure06}.

\section{Conclusion}
\label{sec:discussions}



This work focuses on the canalization structure of Boolean functions, specifically through studying the layer structure of a function. Knowing the layer structure of a Boolean function may aid researchers in multiple ways. It may be used to inform which experiments to prioritize when the goal is the identification of the correct update rule. The hierarchical layer structure may further reveal the level of importance of specific input variables, guiding experimentalists on which aspects of genetic control to focus.

In this paper we developed and implemented two algorithms for finding the layer structure of a general Boolean function, and proposed a third algorithm for the special case of nested canalizing functions.
The implementation in Python requires the truth table format of the function while the Macaulay2 version requires a polynomial version of the function. We showed that the use of the polynomial version is advantageous when the function exhibits canalization; if the function is noncanalizing, the run time of both algorithms increases very similarly when the number of variables increases, at an exponential speed. All code is available at the GitHub repository 
\href{https://github.com/ckadelka/BooleanCanalization}{https://github.com/ckadelka/BooleanCanalization}.
We demonstrated two uses of canalizing layers. In one application given the layer structure of a nested canalizing function, we showed how to find its disjunctive normal form, an \textbf{NP}-hard problem in the general case.  In a second application we showed how to recover the complete family of functions that satisfy given partial layer information.

Improving the space- and time-efficiency of the algorithms for special classes of functions that take advantage of their structure (such as symmetric functions) or developing multi-state implementations may be conducive for increasing the scope of applications.

\bibliographystyle{plain}
\bibliography{ref_for_layers}




\end{document}